\pdfoutput=1
\documentclass[12pt]{amsart}
\usepackage{pdfsync}
\usepackage{enumerate}
\usepackage{graphicx}
\usepackage{latexsym}
\usepackage{amsfonts}
\usepackage[utf8]{inputenc}
\usepackage{amsthm}
\usepackage{amsmath}
\usepackage{amssymb}
\usepackage[
            unicode,
            breaklinks=true,
            colorlinks=true]{hyperref}
\usepackage{amscd}
\usepackage[all]{xy} 
\usepackage{mathrsfs} 
\usepackage{stmaryrd} 
\usepackage{amsopn} 
\usepackage{eepic}
\usepackage{epic}
\usepackage{eucal}
\usepackage{epsfig}
\usepackage{psfrag}

\newtheorem{theorem}{Theorem}[section]

\newtheorem{lemma}[theorem]{Lemma}

\newcommand{\R}{{\mathbb R}}
\newcommand{\Z}{{\mathbb Z}}

\newcommand{\T}{{\mathbb T}}

\newcommand{\ham}[1]{\mathcal{X}_{#1}}

\newcommand{\op}[1]{\!\!\mathop{\rm ~#1}\nolimits}
\newcommand{\DD}{\mathrm{d}}

\renewcommand{\leq}{\leqslant}

\newcommand{\abs}[1]{\left|#1\right|}

\newcommand{\RM}{\mathbb{R}}
\newcommand{\h}{\hbar}

\newenvironment{remark}{\refstepcounter{theorem}\par\medskip\noindent{\bf
Remark~\thetheorem.}}{\unskip\nobreak\hfill\hbox{}}

\newenvironment{definition}{\refstepcounter{theorem}\par\medskip\noindent{\bf
Definition~\thetheorem.}}{\unskip\nobreak\hfill\hbox{}}

\begin{document}

\title[Spectral theory for singularities of
focus\--focus type]{Semiclassical inverse spectral theory for singularities
  of focus\--focus type} \author{\'Alvaro Pelayo
   and San V\~u Ng\d oc} \date{}

\maketitle
\thispagestyle{empty}

\begin{abstract}
  We prove, assuming that the Bohr\--Sommerfeld rules hold, that the
  joint spectrum near a focus\--focus critical
  value of a quantum integrable system determines the classical
  Lagrangian foliation around the full focus\--focus leaf.  The result
  applies, for instance, to $\hbar$\--pseudodifferential operators,
  and to  Berezin\--Toeplitz operators on 
  prequantizable compact symplectic
manifolds. 
\end{abstract}

\section{Introduction}

The development of semiclassical and microlocal analysis  since the
1960's now provides a strong theoretical background to discover new
interactions between spectral theoretic and analytic methods, and
geometric and dynamical ideas from symplectic geometry (see for 
instance Guillemin\--Sternberg~\cite{GuSt2012} and Zworski~\cite{Zw2012}).
The theme of this paper is to study
these interactions for an important class of singularities which appear
in integrable systems: \emph{focus\--focus singularities} (of the associated
Lagrangian foliation). The singular fibers corresponding to such
singularities are pinched tori (Figure~\ref{fig:pinched}).  Focus\--focus singularities appear
naturally also in algebraic geometry in the context of Lefschetz fibrations, 
where they are sometimes called \emph{nodes}.

In this article we consider the joint spectrum of a pair of commuting
semiclassical operators, in the case where the phase space is
four\-dimensional.  We prove, assuming that the Bohr\--Sommerfeld
rules hold, that the joint spectrum in a neighborhood of a
focus\--focus singularity determines the classical dynamics of the
associated system around the focus\--focus fiber. This problem
belongs to a class of inverse spectral questions which has attracted
much attention in recent years, eg. \cite{HeZe2010, CoGu2011}, and which goes back to
pioneer works of Colin de Verdi{\`e}re \cite{CdV, CdV2} and
Guillemin\--Sternberg \cite{GuSt}, in the 1970s and 1980s.

The result applies as soon as one can show that the usual
Bohr\--Sommerfeld rules hold for the quantum system. This includes the
cases of $\hbar$-semiclassical pseudodifferential operators, as shown
in~\cite{charbonnel} and~\cite{vungoc-focus}; it also includes the
interesting case of Berezin-Toeplitz quantization,
see~\cite{charles-quasimodes}.

Examples of quantum integrable systems, given by differential
operators, with precisely one singularity of focus\--focus type are
the spherical pendulum, discussed by Cushman and Duistermaat in
\cite{cushman-duist} (see also~\cite[Chapitre 2]{vungoc-panorama}),
and the ``Champagne bottle''~\cite{child}. Integrable systems given by
Berezin-Toeplitz quantization are also common in the physics
literature.  An important example is the coupling of angular momenta
(spins)~\cite{sadovski-zhilinski} (see also~\cite[Section
8.3]{pelayo-polterovich-vungoc} for a proof that it is indeed a
Berezin-Toeplitz system).

Many other integrable systems have focus\--focus singularities, which
are in fact the simplest integrable singularity with an isolated
critical value. An infinite number of non-isomorphic systems with
focus-focus singularities are provided by the so called
\emph{semitoric systems} constructed in
\cite{pelayo-vungoc-constructing}.

The structure of the paper is as follows: in Section~\ref{ba} we
explain what we mean by a semiclassical operator, and recall the
notion of integrable system in dimension four. In Section~\ref{mt} we
state our main theorem: Theorem~\ref{main}. In
Section~\ref{taylor:sec} we review the construction of the so called
Taylor series invariant, which classifies, up to isomorphisms, a
semiglobal neighborhood of a focus\--focus singularity. In Section
\ref{sec:main} we prove Theorem \ref{main}.

\section{Symplectic theory of integrable systems} \label{ba}

Let $(M,\omega)$ be a smooth, connected $4$\--dimensional symplectic
manifold.

\subsection{Integrable systems}
An \emph{integrable system} $(J,H)$ on $(M,\omega)$ consists of two
Poisson commuting functions $J,H \in \op{C}^{\infty}(M;\mathbb{R})$
i.e.~:
 $$
 \{J,H\}:=\omega(\ham{J},\, \ham{H})=0,
 $$
 whose differentials are almost everywhere linearly independent
 $1$\--forms.  Here $\ham{J},\ham{H}$ are the Hamiltonian
 vector fields induced by $J,H$, respectively, via the symplectic form
 $\omega$: $ \omega(\ham{J},\cdot)=-{\rm d}J$,
 $\omega(\ham{H},\cdot)=-{\rm d}H$.
  
 For instance, let $M_0={\rm T}^*\T^2$ be the cotangent bundle of the
 torus $\T^2$, equipped with canonical coordinates
 $(x_1,x_2,\xi_1,\xi_2)$, where $x\in \T^2$ and $\xi\in {\rm
   T}^*_x\T^2$. The linear system $$(J_0,H_0):=(\xi_1, \xi_2)$$ is
 integrable.

 An \emph{isomorphism} of integrable systems $(J,\,H)$ on
 $(M,\omega)$ and $(J',\,H')$ on $(M',\omega')$ is a
 diffeomorphism $ \varphi \colon M \to M'$ such that
 $\varphi^*\omega'=\omega$ and
$$
\varphi^*(J',\,H')=(f_1(J,\,H),\,f_2(J,\,H))
$$
for some local diffeomorphism $(f_1,f_2)$ of $\R^2$. This same
definition of isomorphism extends to any open subsets $U\subset
M$, $U'\subset M'$ (and this is the form in which we will use it
later). Such an isomorphism will be called \emph{semiglobal} if $U,U'$
are respectively saturated by level sets $\{J=\text{const}_1,
H=\text{const}_2\}$ and $\{J'=\text{const}'_1,
H'=\text{const}'_2\}$.

If $F=(J,H)$ is an integrable system on $(M,\omega)$, consider a point
$c\in\R^2$ that is a \emph{regular value} of $F$, and such that the
fiber $F^{{-1}}(c)$ is compact and connected. Then, by the
action-angle theorem~\cite{duistermaat}, a saturated neighborhood of
the fiber is \emph{isomorphic} in the previous sense to the above
linear model on $M_0={\rm T}^*\T^2$. Therefore, all such regular
fibers (called \emph{Liouville tori}) are isomorphic in a neighborhood.

However, the situation changes drastically when the condition that $c$
be regular is violated. For instance, it has been proved
in~\cite{vungoc-semi-global} that, when $c$ is a so-called
\emph{focus-focus} critical value (see Section~\ref{sec:ff1} below),
an infinite number of equations has to be satisfied in order for two
systems to be semiglobally isomorphic near the critical fiber (see
Section~\ref{taylor:sec}).

\subsection{Focus\--focus singularities} \label{sec:ff1}

Let $\mathcal{F}$ be the \emph{associated singular foliation} to the
integrable system $F=(J,\, H) \colon M \to \R^2$, the leaves of which
are by definition the connected components of the fibers $F^{-1}(c)$.
Let $p$ be a critical point of $F$.  We assume for simplicity that
$F(p)=0$, and that the (compact, connected) fiber
$\Lambda_0:=F^{-1}(0)$ does not contain other critical points.  A
focus\--focus singularity $p$ is characterized by Eliasson's theorem
\cite{eliasson-these,vungoc-wacheux} as follows: there exist
symplectic coordinates $(x,\, y,\, \xi,\,\eta)$ in a neighborhood $W$
around $p$ in which $(q_1,\,q_2)$, given by
\begin{equation}
  q_1=x\eta-y\xi, \,\,   q_2=x\xi+y\eta
  \label{equ:cartan}
\end{equation}
is a momentum map for the foliation $\mathcal{F}$: one has
$F=g(q_1,q_2)$ for some local diffeomorphism $g$ of $\R^2$ defined
near the origin (the critical point $p$ corresponds to coordinates
$(0,\,0,\,0,\,0)$).  One of the major characteristics of focus\--focus
singularities is the existence of a \emph{Hamiltonian action of $S^1$}
that commutes with the flow of the system, in a neighborhood of the
singular fiber that contains $p$~\cite{zung-focus,zung-I}. Such
singularities are also very natural candidates for a topological study
of singular Lagrangian fibrations~\cite{symington-four}.

\section{Main Theorem: inverse spectral theory for focus\--focus
  singularities} \label{mt}

Let $(M,\omega)$ be a $4$\--dimensional connected symplectic manifold.

\subsection{Semiclassical operators}
Let $I \subset (0,1]$ be any set which accumulates at $0$.  If
$\mathcal{H}$ is a complex Hilbert space, we denote by
$\mathcal{L}(\mathcal{H})$ the set of linear (possibly unbounded)
selfadjoint operators on $\mathcal{H}$ with a dense domain. 

A space
$\Psi$ of \emph{semiclassical operators} is a subspace of
$\prod_{\hbar \in I} \mathcal{L}(\mathcal{H}_{\hbar})$ equipped with a
$\R$\--linear map $$\sigma \colon \Psi \to {\rm C}^{\infty}(M;\, \R),$$ called
the \emph{principal symbol map}. If $P=(P_{\hbar})_{\hbar \in I} \in \Psi$, the image $\sigma(P)$ is
called the \emph{principal symbol of $P$}. 

We say that two
semiclassical operators $(P_{\hbar})_{\hbar \in I}$ and
$(Q_{\hbar})_{\hbar \in I}$ \emph{commute} if for each $\hbar \in I$
the operators $P_{\hbar}$ and $Q_{\hbar}$ commute.

\subsection{Semiclassical spectrum}

Let $P=(P_{\hbar})_{\hbar \in I}$ and $Q=(Q_{\hbar})_{\hbar \in I}$ be
semiclassical commuting operators on Hilbert spaces
$(\mathcal{H}_\h)_{\hbar \in I}$, where at each $\hbar \in I$ the
operators have a common dense domain $\mathcal{D}_{\hbar} \subset
\mathcal{H}_{\hbar}$ such that $P_{\hbar}(\mathcal{D}_{\hbar}) \subset
\mathcal{D}_{\hbar}$ and $Q_{\hbar}(\mathcal{D}_{\hbar}) \subset
\mathcal{D}_{\hbar}$. 

For fixed $\hbar$, the \emph{joint spectrum} of
$(P_{\hbar},Q_{\hbar})$ is the support of the joint spectral
measure (see Figure~\ref{fig:spectrum} for an example). It is denoted by $\op{JointSpec}(P_{\hbar},\,Q_{\hbar})$. If
$\mathcal{H}_\h$ is finite dimensional, then
$$
\op{JointSpec}(P_{\hbar},\,Q_{\hbar})=\Big\{(\lambda_1,\lambda_2)\in
\R^2\,\, |\,\, \exists v\neq 0,\,\, P_{\hbar} v = \lambda_1
v,\,\,Q_{\hbar} v = \lambda_2 v \Big\}.
 $$
 The \emph{joint spectrum} of $P,Q$ is the collection of all joint
 spectra of $(P_{\hbar},Q_{\hbar})$, $\hbar \in I$. It is denoted by
 $\op{JointSpec}(P,\,Q)$. For convenience of the notation, we will
 also view the joint spectrum of $P,Q$ as a set depending on $\hbar$.

\subsection{Bohr\--Sommerfeld rules}
Recall that the \emph{Hausdorff distance} between two subsets $A$ and
$B$ of $\R^2$ is
 $$
 {\rm d}_H(A,\,B):= \inf\{\epsilon > 0\,\, | \,\ A \subseteq
 B_\epsilon \ \mbox{and}\ B \subseteq A_\epsilon\},
$$
where for any subset $X$ of $\R^2$, the set $X_{\epsilon}$ is
$$X_\epsilon := \bigcup_{x \in X} \{m \in \R^2\, \, | \,\, \|x - m \|
\leq \epsilon\}.$$ If $(A_{\hbar})_{\hbar \in I}$ and
$(B_{\hbar})_{\hbar \in I}$ are sequences of subsets of
$\mathbb{R}^2$, we say that $$A_{\hbar} = B_{\hbar} +
\mathcal{O}(\hbar^{N}) $$ if there exists a constant $C>0$ such that 
$$
{\rm
  d}_H(A_{\hbar},\,B_{\hbar})\leq C\hbar^{N}$$ 
  for all $\hbar \in I$.

\begin{figure}[htbp]
  \begin{center}
    \includegraphics[width=6cm]{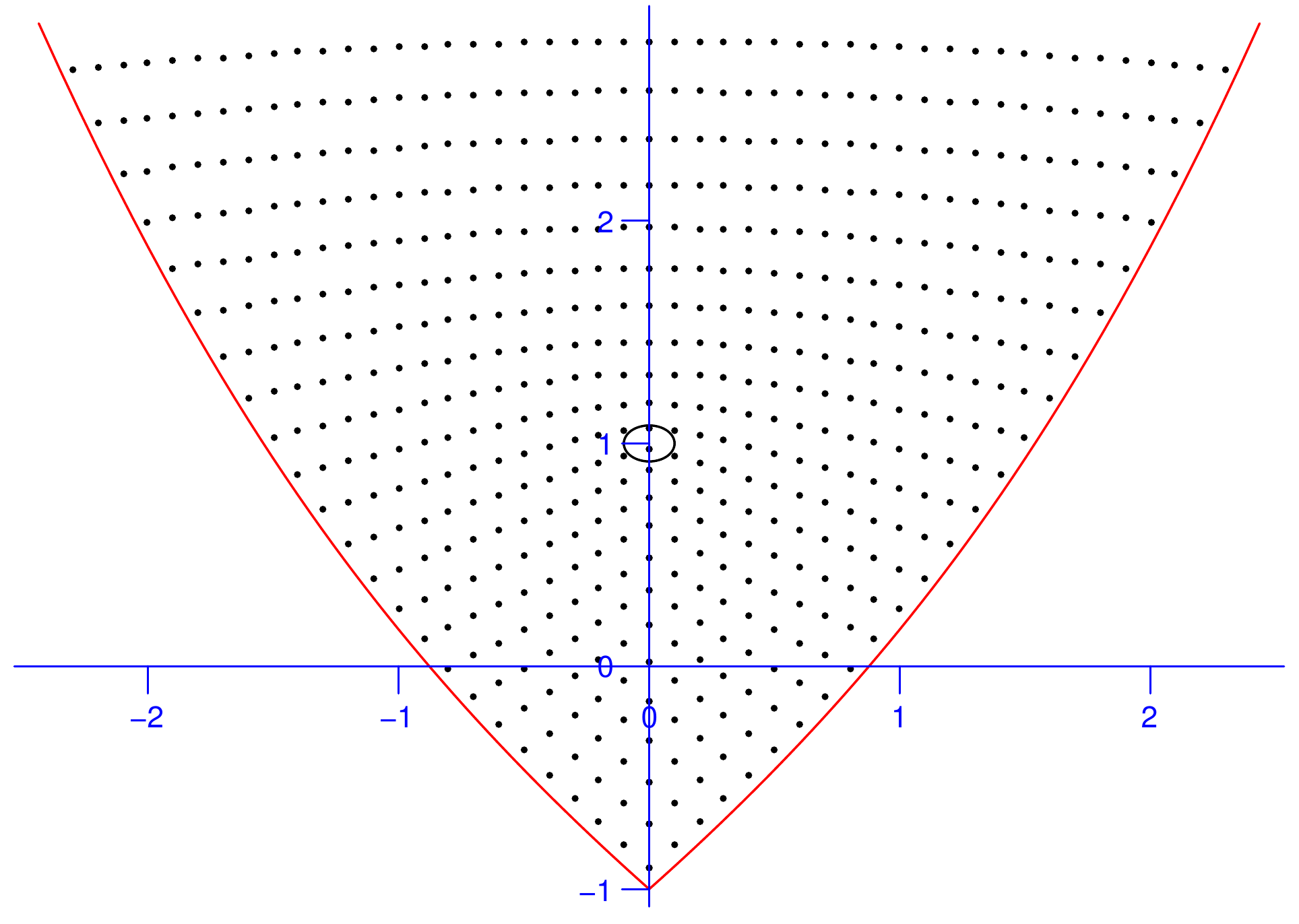}
    \caption{Joint spectrum of the quantum spherical pendulum when
      $\hbar=\frac{1}{10}$ (figure taken
      from~\cite{vungoc-panorama}).}
    \label{fig:spectrum}
  \end{center}
\end{figure}

\begin{definition} \label{chart} Let $F:=(J,H) \colon M \to
  \mathbb{R}^2$ be an integrable system on a $4$\--dimensional
  connected symplectic $4$\--manifold.  Let $P$ and $Q$ be commuting
  semiclassical operators with principal symbols $J,H \colon M \to
  \mathbb{R}$.  Let $U \subset \mathbb{R}^2$ be an open set. We say
  that $\op{JointSpec}(P,\,Q)$ \emph{satisfies the Bohr\--Sommerfeld
    rules on} $U$ if for every regular value $c \in U$ of $F$ there
  exists a small ball ${\rm B}(c,\epsilon_c)\subset U$ centered at $c$
  such that
  \begin{eqnarray} \label{bs1} \nonumber \op{JointSpec}(P,\,Q) \cap
    {\rm B}(c,\epsilon_c)= g_{\hbar}(2\pi \hbar \Z^2\cap D) \cap {\rm
      B}(c,\epsilon_c) + \mathcal{O}(\hbar^2)
  \end{eqnarray}
  with $$g_{\hbar}=g_0+\hbar g_1,$$ where $g_0,g_1$ are smooth maps
  defined on a bounded open set $D\subset\R^2$, $g_0$ is a
  diffeomorphism into its image, and the components of
  $g_0^{-1}=(\mathcal{A}_1,\,\mathcal{A}_2)$ form a basis of action
  variables.  We call $(g_{\hbar})^{-1}$ an \emph{affine chart} for
  $\op{JointSpec}(P,\,Q)$.
\end{definition}

\smallskip For instance, Bohr\--Sommerfeld rules are known to hold for
pseudodifferential operators on cotangent bundles
\cite{charbonnel,vungoc-focus}, or for Toeplitz operators on
prequantizable compact symplectic
manifolds~\cite{charles-quasimodes}. It would be interesting to
formalize the minimal semiclassical category where Bohr-Sommerfeld
rules are valid.

\begin{remark} \label{rk} If $(g_{\hbar})^{-1}$ is an affine chart for
  ${\rm JointSpec}(P,Q)$ and $B \in {\rm GL}(2,\Z)$ then $B \circ
  (g_{\hbar})^{-1}$ is again an affine chart.
\end{remark}

\subsection{Main Theorem}

Let ${\rm CIS}(M,\omega)$ be the set of \emph{classical integrable
  systems}
\[
 F=(J,H)\colon M \to \mathbb{R}^2
\]
on the connected $4$\--dimensional symplectic manifold $(M,\omega)$, such that
$F$ is a proper map.  Let ${\rm QIS}(M,\omega)$ be the set of
\emph{quantum integrable systems} given as pairs of commuting
semiclassical operators $(P,Q)$ whose principal symbols, say
$(J,H)=\sigma(P,Q)$, form an integrable system $F\in {\rm
  CIS}(M,\omega)$. Let $\mathcal{P}(\R^2)$ the set of subsets of
$\R^2$, and consider the following diagram:

\begin{eqnarray} \label{mainlemma} \xymatrix@C=5em{ 
    {\rm QIS}(M,\omega) \ar[r]^>>>>>>>>>>{{\rm JointSpec}} \ar[d]^{
      \sigma } & \mathcal{P}(\R^2)
    \\
    {\rm CIS}(M,\omega) & } \nonumber
\end{eqnarray}

\begin{theorem} \label{main} Let $(M,\omega)$ and $(M',\omega')$ be
  connected $4$\--dimensional symplectic manifolds. Let $(P,Q)$ and
  $(P,Q')$ be quantum integrable systems on $(M,\omega)$ and
  $(M',\omega')$, respectively, which have a focus\--focus singularity
  at points $p\in M$ and $p'\in M'$ respectively.  Suppose that
  $c_0:=\sigma(P,Q)(p)=\sigma(P',Q')(p')$ and that there exists a
  neighborhood $U$ of $c_0$ such that ${\rm JointSpec}(P,Q)$ and ${\rm
    JointSpec}(P',Q') $ satisfy the Bohr\--Sommerfeld rules on $U$. If
 $$
 {\rm JointSpec}(P,Q)\cap U=({\rm JointSpec}(P',Q')\cap U) + \mathcal{O}(\hbar^2),
 $$
 then there are saturated neighborhoods $\mathcal{V},\mathcal{V}'$ of
 the singular fibers of $\sigma(P,Q)$ and $\sigma(P',Q')$
 respectively, such that the restrictions $\sigma(P,Q)|_{\mathcal{V}}$
 and $\sigma(P',Q')|_{\mathcal{V'}}$ are isomorphic as integrable
 systems.
\end{theorem}

\section{Taylor series invariant at focus\--focus singularity}
\label{taylor:sec}

We use here the notation of Section \ref{sec:ff1}. In particular $p\in
M$ is a focus-focus point, and $W$ is a small neighborhood of $p$. Fix
$A'\in \Lambda_0\cap (W\setminus\{p\})$ and let $\Sigma$ denote a
small 2\--dimensional surface transversal to $\mathcal{F}$ at the
point $A'$.  Since the Liouville foliation in a small neighborhood of
$\Sigma$ is regular for both $F$ and $q=(q_1,\,q_2)$, there is a
diffeomorphism $\varphi$ from a neighborhood $U$ of $F(A')\in\R^2$
into a neighborhood of the origin in $\R^2$ such that $q=\varphi \circ
{F}$. Thus there exists a smooth momentum map $\Phi=\varphi \circ{F}$
for the foliation, defined on a neighborhood $\Omega=F^{-1}(U)$ of
$\Lambda_0$, which agrees with $q$ on $W$.  Write $\Phi:=(H_1,\,H_2)$
and $\Lambda_c:=\Phi^{-1}(c)$.  Note that $ \Lambda_0=\mathcal{F}_p.
$ It follows from~(\ref{equ:cartan}) that near $p$ the $H_1$\--orbits
must be periodic of primitive period $2\pi$, whereas the vector field
$\ham{H_2}$ is hyperbolic with a local stable manifold (the
$(\xi,\eta)$-plane) transversal to its local unstable manifold (the
$(x,y)$-plane). Moreover, $\ham{H_2}$ is \emph{radial}, meaning that
the flows tending towards the origin do not spiral on the local
(un)stable manifolds.
  
Suppose that $A \in\Lambda_c$ for some regular value $c$.
  
\begin{definition} \label{taus} Let $\tau_2(c)>0$ be the smallest time
  it takes the Hamiltonian flow associated with $H_2$ leaving from $A$
  to meet the Hamiltonian flow associated with $H_1$ which passes
  through $A$.  Let $\tau_1(c)\in[0,2\pi)$ be the time that it takes to
  go from this intersection point back to $A$, closing the trajectory.
\end{definition}
  
In Definition \ref{taus}, the existence of $\tau_2$ is ensured by the
fact that the flow of $H_2$ is a quasiperiodic motion always
transversal to the $S^1$-orbits generated by $H_1$.

The commutativity of the flows ensure that $\tau_1(c)$ and $\tau_2(c)$
do not depend on the initial point $A$.
\begin{figure}[htbp]
  \begin{center}
    \includegraphics[width=6cm]{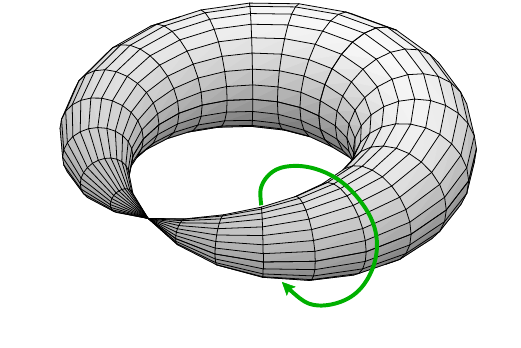}
    \caption{Fiber containing a focus\--focus singularity.}
    \label{fig:pinched}
  \end{center}
\end{figure}
Write $c=(c_1,\,c_2)=c_1+\op{i}c_2$ ($c_1,\, c_2 \in \R$), and let
$\op{log} c$ be a fixed determination of the logarithmic function on
the complex plane. Let
\begin{equation}
  \left\{
    \begin{array}{ccl}
      \sigma_1(c) & = & \tau_1(c)-{\rm Im}(\op{log} c) \\
      \sigma_2(c) & = & \tau_2(c)+{\rm Re}(\op{log} c),
    \end{array}
  \right.
  \label{equ:sigma} 
\end{equation}
where ${\rm Re}$ and ${\rm Im}$ respectively stand for the real and imaginary
parts of a complex number.  V\~u Ng\d oc proved in
\cite[Proposition~\,3.1]{vungoc-semi-global} that $\sigma_1$ and
$\sigma_2$ extend to smooth and single\--valued functions in a
neighborhood of $0$ and that the differential 1\--form
$\sigma:=\sigma_1\, \DD{}c_1+\sigma_2\, \DD{}c_2 $ is closed.  Notice
that if follows from the smoothness of $\sigma_1$ that one may choose
the lift of $\tau_1$ to $\R$ such that
$\sigma_1(0)\in[0,\,2\pi)$. This is the convention used throughout.

Following \cite[Definition~3.1]{vungoc-semi-global} , let $S$ be the
unique smooth function defined around $0\in\R^2$ such that
\begin{eqnarray}
  \DD{}S=\sigma,\,\, \,\, S(0)=0.
\end{eqnarray}
The Taylor expansion of $S$ at $(0,\,0)$ is denoted by $(S)^\infty$.

\begin{definition}
  The expansion $(S)^{\infty}$ is a formal power series in two
  variables with vanishing constant term, and we call it the \emph{Taylor series invariant of $(J,\,H)$
    at the focus\--focus point $c_0$}.
\end{definition}

\begin{theorem}[\cite{vungoc-semi-global}] \label{theo:vn} The Taylor
  series invariant $(S)^{\infty}$ charaterizes, up to symplectic
  isomorphisms, a semiglobal saturated neighborhood of the singular
  fiber of the focus\--focus singularity $p$.
\end{theorem}

It is interesting to notice that, in the famous case of the spherical
pendulum, the Taylor series invariant was recently explicitly
computed~\cite{dullin-pendulum}.

\section{Proof of Theorem \ref{main}} \label{sec:main}

In view of Theorem~\ref{theo:vn}, we wish to prove that the symplectic
invariant $(S)^\infty$ is determined by the joint spectrum. The proof
is organized in several statements. Throughout we use the notation of
Section \ref{taylor:sec}. Let $(P,Q)\in {\rm QIS}(M,\omega)$ and
$(P,Q')\in{\rm QIS}(M',\omega')$ be quantum integrable systems which
have a focus\--focus singularity at points $p\in M$ and $p'\in M'$,
respectively.  Suppose that $c_0:=\sigma(P,Q)(p)=\sigma(P',Q')(p')$
and that there exists a neighborhood $U$ of $c_0$ such that ${\rm
  JointSpec}(P,Q)$ and ${\rm JointSpec}(P',Q') $ satisfy the
Bohr\--Sommerfeld rules on $U$.  Assume that
 $$
 {\rm JointSpec}(P,Q)\cap U=({\rm JointSpec}(P',Q')\cap U) +
 \mathcal{O}(\hbar^2)
 $$
 Let $\sigma(P,Q)=(J,H):=F \colon M \to \mathbb{R}^2$ and let
 $\sigma(P',Q')=(J',H'):=F' \colon M' \to \mathbb{R}^2$.

 \begin{lemma}{\rm
     (\cite[Proposition~1]{vungoc-monodromy})} \label{pro} Let
   $(f_{\hbar})^{-1}$ and $(g_{\hbar})^{-1}$ be two affine charts for
   ${\rm JointSpec}(P,Q)$, both defined on a ball $B$ around $c$. Let
   $(f_0)^{-1},(g_0)^{-1}$ denote the principal symbols of
   $(f_{\hbar})^{-1}$ and $(g_{\hbar})^{-1}$ respectively.  Then there
   is a constant matrix $C\in {\rm GL}(2,\Z)$ such that for all $c\in
   B$, we have that $${\rm d}(g_0)^{-1}(c)= C\cdot({\rm
     d}(f_0)^{-1}(c)).$$
 \end{lemma}

\begin{proof}
  We recall the proof for the reader's convenience. Let
  $\Sigma_\h:={\rm JointSpec}(P_{\hbar},Q_{\hbar})$.  Let $c\in U$, and
  $(f_{\hbar})^{-1}$, $(g_{\hbar})^{-1}$ be two affine charts of
  $\Sigma$ defined on a ball $B$ around $c$. Any open ball around $c$
  contains, for $\hbar$ small enough, at least one element of
  $\Sigma_{\hbar}$. Therefore, there exists a family
  $\lambda_{\hbar}\in \Sigma_{\hbar}\cap B$ such that $\lim_{\hbar\to
    0} \lambda_{\hbar} = c$.  Let $k\in\Z^2$ and let
  $\lambda'_{\hbar}$ be a family of elements of $\Sigma_{\hbar}\cap B$
  such that
  \[
  (f_{\hbar})^{-1}(\lambda_{\hbar}) =
  (f_{\hbar})^{-1}(\lambda'_{\hbar}) + \hbar k + \mathcal{O}(\hbar^2).
  \]
  Then, as $\hbar$ tends to zero,
  $\frac{\lambda'_{\hbar}-\lambda_{\hbar}}{\hbar}$ tends towards a
  limit $v\in\RM^2$ which satisfies $k = {\rm d}(f_0)^{-1}(c)v$. Since
  $\lambda_{\hbar}$ and $\lambda'_{\hbar}$ are in $\Sigma_{\hbar}$,
  there is a family $k'_{\hbar}\in\Z^2$ such that
  \begin{eqnarray} \label{ab}
  \left(\frac{(g_{\hbar})^{-1}(\lambda'_{\hbar})-(g_{\hbar})^{-1}(\lambda_{\hbar})}{\hbar}\right)
  = k'_{\hbar} + \mathcal{O}(\hbar).
  \end{eqnarray}
  The left-hand side of (\ref{ab}) has limit ${\rm
    d}(g_0)^{-1}(c)v$ as $\hbar\to 0$. Therefore $k'_{\hbar}$ is equal
  to a constant integer $k'$ for small $\hbar$, and we have $k' = {\rm
    d}(g_0)^{-1}(c)({\rm d}(f_0)^{-1}(c))^{-1}k$, which implies
  that $${\rm d}(g_0)^{-1}(c)({\rm d}(f_0)^{-1}(c))^{-1}\in {\rm
    GL}(2,\Z).$$ Since ${\rm GL}(2,\Z)$ is discrete, the conclusion of
  the lemma follows.
\end{proof}
 
Next we proceed in several steps.

\vskip 1em
\noindent\emph{Step 1}. First we normalize the systems $F:=(J,H)
\colon M \to \R^2$ and $F':=(J',H') \colon M' \to \mathbb{R}^2$ at the
focus\--focus singular points, which doesn't change the Taylor series
invariant. In order to do this, let $\varphi \colon \Omega \to
(\R^4,\omega_0)$ be a symplectomorphism into its image, where $\Omega$
is a neighborhood of the singular point $p$ and $\omega_0$ is the
standard symplectic form on $\R^4$, such that $ F\circ
\varphi^{-1}=g(q_1,q_2) $ near $(0,0,0,0)$ (it exists by Eliasson's
Theorem, see Section~\ref{sec:ff1}).  Here $g$ is some local
diffeomorphism.  Similarly let $\varphi' \colon \Omega' \to
(\R^4,\omega_0)$ be a local symplectomorphism, where $\Omega'$ is a
small neighborhood of the singular point $p'$ such that $ F'\circ
(\varphi')^{-1}=g'(q_1,q_2).  $ Here $g'$ is some local
diffeomorphism.

By replacing $F$ by the integrable system $g^{-1} \circ F$ and $F'$ by
$(g')^{-1} \circ F'$, defined respectively on semiglobal neighborhoods
$\mathcal{V}$, $\mathcal{V}'$ we may assume that $c_0=0$ and that $g$
and $g'$ are both the identity near the origin, that is, we may assume
that
 $$
 F\circ \varphi^{-1}=(q_1,q_2),\,\,\,\,\,\,\,\,\,\, F'\circ
 (\varphi')^{-1}=(q_1,q_2)
 $$
 near $(0,0,0,0)$.
 \\
 \\
 \emph{Step 2}. Denote by ${\rm B}_r$ the set of regular values of $F$
 and $F'$ simultaneously.  Since the focus-focus critical value
 $c_0=0$ is isolated, there exists a small ball $U$ around $0$ such
 that $\dot{U}:=U \setminus \{0\}={\rm B}_r\cap U$.

 Let $c\in\dot U$; let $\Lambda_c:=F^{-1}(c)$, which is a Liouville
 torus. Let $(\delta_1,\delta_2)$ be loops in $\Lambda_c$ that form a
 basis of cycles in $H_1(\Lambda_c,\R)$. Let
 $$
 \mathcal{A}_j(c)=\int_{\delta_j(c)} \nu,\,\,\,\,\,\, j=1,2
 $$
 be the action integrals, where $\nu$ is a $1$\--form such that $\DD
 \nu=\omega$.  Similarly, let $\Lambda'_c:=(F')^{-1}(c)$, let
 $(\delta'_1,\delta'_2)$ be a basis of cycles in $\Lambda'_c$, and let
 $$
 \mathcal{A}_j(c)=\int_{\delta'_j(c)} \nu',\,\,\,\,\,\, j=1,2.
 $$
 Write $ \mathcal{A}:=(\mathcal{A}_1,\mathcal{A}_2)$ and
 $\mathcal{A}':=(\mathcal{A}'_1,\mathcal{A}'_2)$. It follows from the
 action-angle theorem that $\mathcal{A}$ and $\mathcal{A}'$ are local
 diffeomorphisms of $\R^2$ defined near $c$.

 \begin{lemma} \label{52} There exists a matrix $B \in {\rm GL}(2,\Z)$
   such that we have the following relation between the integrals
   above~:
   \begin{eqnarray} \label{B:matrix} \mathcal{A}(c)=B \circ
     \mathcal{A}'(c) + {\rm constant}, \nonumber
   \end{eqnarray}
   for all $c \in \dot U$.
 \end{lemma}

\begin{proof}
  Fix $c\in \dot U$.  Let $(g_{\hbar})^{-1},(g'_{\hbar})^{-1} $ be
  affine charts near $c$ for the spectra ${\rm JointSpec}(P,Q)$ and
  ${\rm JointSpec}(P',Q')$ respectively. By Remark \ref{rk} we may
  assume that $(g_0)^{-1}=(\mathcal{A}_1,\mathcal{A}_2)$ and
  $(g'_0)^{-1}=(\mathcal{A}'_1,\mathcal{A}'_2)$.  Since the joint
  spectra are equal modulo $\mathcal{O}(\hbar^2)$, $(g'_{\hbar})^{-1}$
  is also an affine chart for ${\rm JointSpec}(P,Q)$.  Therefore by
  Lemma \ref{pro} there is a constant matrix $B \in {\rm GL}(2,\Z)$
  such that for all $c'$ near $c$,
  \begin{eqnarray} \label{ds:eq} {\rm d}(g_0)^{-1}(c')= B \cdot({\rm
      d}(g'_0)^{-1}(c')). \nonumber
  \end{eqnarray} 
  Since $B$ is constant in a neighborhood of $c$, it does not depend
  on $c\in \dot U$, which proves the lemma.
\end{proof}
By replacing $(\delta'_1,\delta'_2)$ by
$(B^{-1})(\delta'_1,\delta'_2)$ we may assume that the matrix $B$ in
(\ref{B:matrix}) is the identity matrix.

\vskip 1em
\noindent\emph{Step 3}. Consider the Hamiltonian vector field
$\ham{{}J}$. Notice that $J$ is a momentum map for an $S^1$\--action
on $U$ (recall that $J \circ \varphi^{-1}=q_1$ on $\Omega$).  Recall
the times $\tau_1,\tau_2$ in Definition \ref{taus}.  Let $\gamma_1(c)$
be the $2\pi$\--periodic orbit of $\ham{J}$ and let $\gamma_2(c)$ be
the loop constructed as the flow of the vector field $\tau_1
\ham{J}+\tau_2\ham{H}$.  The pair $(\gamma_1,\gamma_2)$ is a basis of
the homology group ${\rm H}_1(\Lambda_c,\R)$. Similarly define
$\ham{J'}$, $\ham{H'}$, $\tau'_1, \tau'_2$, and
$(\gamma'_1,\gamma'_2)$. We have that 
$$
(\gamma_1,\gamma_2)=C(\delta_1,\delta_2)$$ and 
$$(\gamma'_1,\gamma'_2)=C'(\delta'_1,\delta'_2)$$ for some matrices
$C,\,C' \in {\rm GL}(2,\Z)$.

\vskip 1em
\noindent\emph{Step 4}.  Let $\mathcal{I}_1,\mathcal{I}_2$ be the
actions corresponding to $\gamma_1,\gamma_2$. Then by Lemma~\ref{52}
we have that
\begin{eqnarray} \label{m} {\rm d}(\mathcal{I}_1,\mathcal{I}_2)=C{\rm
    d}(\mathcal{A}_1,\mathcal{A}_2)=C{\rm
    d}(\mathcal{A}'_1,\mathcal{A}'_2)= C \cdot (C')^{-1} {\rm
    d}(\mathcal{I}'_1,\mathcal{I}'_2). \nonumber
\end{eqnarray}
We want to show that $C=C'$. We write
$$
{\rm d}(\mathcal{I}_1,\mathcal{I}_2)= \left(
  \begin{array}{cc}
    a & b\\ \alpha & \beta
  \end{array}
\right) {\rm d} (\mathcal{I}'_1,\mathcal{I}'_2),
$$
where $a,b,\alpha,\beta\in\Z$ are constant.  The actions
$\mathcal{I}_1,\mathcal{I}'_1$ are well\--defined on $\dot{U}$ because
they come from the $S^1$\--action.  But $\mathcal{I}_2,\mathcal{I}'_2$
are not single\--valued on $\dot{U}$ because of monodromy (this
follows from~\eqref{equ:sigma}).  We have that:
\begin{eqnarray} \label{eq:i1} {\rm d} \mathcal{I}_1=a{\rm
    d}\mathcal{I}'_1+b{\rm d}\mathcal{I}'_2.
\end{eqnarray}
Hence $b=0$ (since otherwise $\mathcal{I}'_2$ would be
single\--valued), and
\begin{eqnarray} \label{eq:i2} {\rm d}\mathcal{I}_2=\alpha{\rm
    d}\mathcal{I}'_1+\beta{\rm d}\mathcal{I}'_2.
\end{eqnarray}
Since $\left(
  \begin{array}{cc}
    a & b\\ \alpha & \beta
  \end{array}
\right) \in {\rm GL}(2,\Z)$ we must have that
$$
\left(
  \begin{array}{cc}
    a & b\\ \alpha & \beta
  \end{array}
\right) =\left(
  \begin{array}{cc}
    \pm 1 & 0\\ \alpha & \pm1
  \end{array}
\right).
$$

\vskip 1em
\noindent\emph{Step 5}. We will use the following well\--known result,
see for instance \cite{duistermaat,vungoc-focus}.

\begin{lemma} \label{lem3} Let $\mathcal{I}=\int_{\gamma_{\alpha}}
  \nu$ where ${\rm d}\nu=\omega$, $F=(J,H) \colon M \to \R^2 $ be a
  proper integrable system on a connected symplectic $4$\--manifold
  $(M,\omega)$, and $\gamma_c \subset \Lambda_c=F^{-1}(c)$ is a smooth
  family of loops drawn on $\Lambda_c$, where $c$ varies in a small
  ball of regular values of $F$. Let $c\mapsto x(c)$ and $y\mapsto
  y(c)$ be smooth functions such that the oriented loop $\gamma_c$ is
  homologous to the $[0,1]$\--orbit of the flow of the vector field
  $x\ham{J}+y\ham{H}$. Then $ {\rm d}\mathcal{I}=x{\rm d}c_1+y{\rm
    d}c_2.  $
\end{lemma}

{F}rom Lemma \ref{lem3} we obtain
$$
{\rm d}\mathcal{I}_2=\tau_1{\rm d}c_1+\tau_2{\rm d}c_2,\,\,\,\,\, {\rm
  d}\mathcal{I}'_2=\tau'_1{\rm d}c_1+\tau'_2{\rm d}c_2,
$$
and $ {\rm d}\mathcal{I}_1={\rm d}\mathcal{I}'_1=2\pi\, {\rm d}c_1.  $
{F}rom equation (\ref{eq:i1}) we get
$$
{\rm d}c_1=a{\rm d}c_1
$$
and hence $a=1$.  From equation (\ref{eq:i2}) we get that
$$
\tau_1{\rm d}c_1 +\tau_2{\rm d}c_2=2\pi \alpha {\rm
  d}c_1+\beta(\tau'_1{\rm d}c_1+\tau'_2{\rm d}c_2),
$$
and therefore
\[
\begin{cases}
  \tau_1= 2\pi \alpha+\beta \tau'_1,\qquad 0 \leq \tau_1<2\pi; \quad 0\leq \tau'_1<2\pi;\\
  \tau_2=\beta \tau'_2
\end{cases}
\]
Hence $$2\pi \abs{\alpha} =| \tau_1-\tau'_1| <2\pi,$$ which implies
$\alpha=0$.  On the other hand, it follows from (\ref{equ:sigma}) that
$$
\sigma_2-{\rm Re}({\rm log}c)=\beta (\sigma_2'-{\rm Re}({\rm log}c)).
$$
Hence $$(1-\beta){\rm Re}({\rm log}c)=\sigma_2-\beta \sigma'_2,$$ which
implies, since $\sigma_2$ and $\sigma'_2$ extend smoothly to $U$, that
$\beta=1$. Therefore we have proven:
$$
\left(
  \begin{array}{cc}
    a & b\\ \alpha & \beta
  \end{array}
\right) =\left(
  \begin{array}{cc}
    1 & 0\\ 0 & 1
  \end{array}
\right).
$$
It follows that there exist some constants $k_1,k_2$ such that
$\mathcal{I}_1=\mathcal{I}'_1+k_1$ and $\mathcal{I}_2=\mathcal{I}'_2 +
k_2$ on $\dot U$.

\vskip 1em
\noindent\emph{Step 6}. Now from~\eqref{equ:sigma} (see
also~\cite{vungoc-semi-global}) we deduce that there are constants
$K,K'$ such that
\begin{eqnarray} \label{star}
  \begin{cases}
    \mathcal{I}_2(c)=K-{\rm Re}(c{\rm log}c-c)+S(c)\\
    \mathcal{I}'_2(c)=K'-{\rm Re}(c {\rm log}c-c)+S'(c) \\
  \end{cases}
\end{eqnarray}
Hence $${\rm d}S(c)={\rm d}S'(c),$$ and therefore the symplectic
invariants $(S)^{\infty}$ and $(S')^{\infty}$ are equal. The result
now follows from Theorem~\ref{theo:vn}.

\vskip 3em

{\emph{Acknowledgements}.  This paper was written at the Institute for
  Advanced Study.   AP was partially
supported by NSF Grants DMS-0965738 and DMS-0635607, an NSF CAREER
Award, a Leibniz Fellowship, Spanish Ministry of Science Grants MTM
2010-21186-C02-01 and Sev-2011-0087.   VNS is partially supported by the Institut
  Universitaire de France, the Lebesgue Center (ANR Labex LEBESGUE),
  and the ANR NOSEVOL grant.  He gratefully acknowledges the
  hospitality of the IAS. }

\bibliographystyle{abbrv}%
\bibliography{bibliography}

\newpage

\noindent
\medskip\noindent

\noindent
\noindent
\\
{\bf {\'A}lvaro Pelayo} \\
School of Mathematics\\
Institute for Advanced Study\\
Einstein Drive\\
Princeton, NJ 08540 USA.
\\
\\
\noindent
Washington University in St Louis\\ 
Mathematics Department \\
One Brookings Drive, Campus Box 1146\\
St Louis, MO 63130-4899, USA.\\
{\em E\--mail}: \texttt{apelayo@math.wustl.edu},  \texttt{apelayo@math.ias.edu}  \\
{\em Website}: \url{http://www.math.wustl.edu/~apelayo/}

\medskip\noindent

\noindent
\noindent
{\bf San V\~u Ng\d oc} \\
Institut Universitaire de France
\\
\\
Institut de Recherches Math\'ematiques de Rennes\\
Universit\'e de Rennes 1\\
Campus de Beaulieu\\
F-35042 Rennes cedex, France\\
{\em E-mail:} \texttt{san.vu-ngoc@univ-rennes1.fr}\\
{\em Website}: \url{http://blogperso.univ-rennes1.fr/san.vu-ngoc/}

\end{document}